\documentclass[aps,pra,10pt,twocolumn,superscriptaddress,floatfix,nofootinbib,showpacs,longbibliography]{revtex4-2}

\usepackage{enumitem}

\usepackage{mathrsfs}
\usepackage[mathscr]{eucal}

\usepackage[utf8]{inputenc}
\usepackage[T1]{fontenc}     
\usepackage[british]{babel}  
\usepackage[sc,osf]{mathpazo}
\usepackage[colorlinks=true, citecolor=purple, urlcolor=blue]{hyperref}

\usepackage{fontawesome5}
\usepackage{eso-pic}
\usepackage{charter}
\usepackage[normalem]{ulem}

\usepackage{graphicx} 
\usepackage[babel]{microtype}  
\usepackage{amsmath,amssymb,amsthm,bm,amsfonts,mathrsfs,bbm} 
\usepackage{natbib}
\usepackage{xcolor,colortbl}
\usepackage{color}
\newcommand{\Tr}{\operatorname{Tr}}
\newcommand{\ketbra}[2]{|#1\rangle \langle #2|}
\newtheorem{theorem}{Theorem}
\newtheorem{definition}{Definition}

\newtheorem{lemma}{Lemma}

\def\01{\{0,1\}}
\newcommand{\norm}[1]{{\left\|{#1}\right\|}}

\newcommand{\E}{\mathbb{E}}

\providecommand{\ket}[1]{| #1{\rangle}}
\providecommand{\bra}[1]{\langle #1|}

\begin{document}

\title{Exponential Advantage of Multipartite Entanglement over Quantum Communication with Applications to Bounded-Storage Cryptography}

\author{Ananya Chakraborty}
\email{ananyaphys.c@gmail.com }
\affiliation{S. N. Bose National Centre for Basic Sciences, Block JD, Sector III, Salt Lake, Kolkata 700 106, India.}

\author{Manik Banik}
\affiliation{S. N. Bose National Centre for Basic Sciences, Block JD, Sector III, Salt Lake, Kolkata 700 106, India.}

\author{Ronald de Wolf}
\affiliation{QuSoft, CWI and University of Amsterdam, the Netherlands. }

\begin{abstract}
We establish an exponential communication advantage enabled by multipartite quantum entanglement. Building on the bipartite Hidden Matching problem, we introduce a communication task involving multiple spatially separated senders and a single receiver. We show that a shared Greenberger–Horne–Zeilinger state enables completion of this task using only logarithmically many bits of classical communication from each sender. In contrast, without preshared entanglement, any protocol achieving high success probability requires polynomial communication from at least one sender, even when \emph{quantum} communication is allowed. Thus, classical communication assisted by multipartite entanglement can be exponentially more powerful than quantum communication without preshared entanglement. As a cryptographic application, we construct a seeded two-source randomness extractor and establish an exponential separation between entangled and unentangled quantum side-information. Specifically, compromising the extractor with two unentangled quantum states storing information about the two sources, respectively, requires polynomial-size memory, whereas exponentially smaller quantum memory suffices in the presence of a small amount of shared entanglement. 

\end{abstract}

\maketitle	

\section{Introduction} 

Quantum entanglement, `\emph{the characteristic trait of quantum mechanics}' \cite{Schrodinger1935}, underlies some of the most profound phenomena in physics, including the Einstein-Podolsky-Rosen paradox, Schr\"{o}dinger steering, and Bell nonlocality \cite{Einstein1935,Bohr1935,Schrodinger1936,Bell1964,Bell1966}. Beyond its foundational significance, entanglement is a key resource for quantum information processing, enabling communication, computation, and cryptographic tasks that are impossible or substantially less efficient within classical frameworks \cite{Bennett1997,Steane1998,Bennett2000,Gisin2002,Gisin2007,Horodecki2009,Buhrman2010,Childs2010,Huang2026}. The communication advantages afforded by bipartite entanglement are now reasonably well understood \cite{Wilde2017}. By contrast, \emph{multipartite} entanglement possesses a far richer structure, including genuinely multipartite forms that cannot be reduced to collections of bipartite correlations \cite{Svetlichny1987,Bouwmeester1999,Dur2000}. Despite remarkable experimental advances in generating and controlling multipartite entangled states \cite{Riedel2010,Wang2018,Figgatt2019,Omran2019,Mooney2021}, our understanding of how such states enhance distributed information-processing tasks remains comparatively limited. Communication complexity provides a natural framework for investigating this question. In communication complexity tasks, spatially separated parties seek to evaluate a global property of distributed inputs while exchanging as little communication as possible \cite{Yao1979,Kushilevitz1996}. Quantum resources can dramatically reduce communication costs and, in several notable instances, yield exponential advantages over classical protocols \cite{Yao1993,Cleve1997,deWolf2002,Brassard2003(1)}. However, most known exponential separations are in bipartite settings \cite{Buhrman1998,Raz1999,Buhrman2001(1),BarYossef2004,Gavinsky2007,Kumar2019}, leaving open the question whether multipartite entanglement alone can provide comparable advantages in  communication tasks.

Here, we address this question by introducing a multipartite one-way communication task wherein multiple spatially separated senders receive classical inputs and a single receiver must produce an output satisfying a prescribed global relation. We show that a Greenberger-Horne-Zeilinger (GHZ) state \cite{Greenberger1990,Mermin1990,Pan2000} shared among all the parties enables success using only $O(\log n)$ bits of \emph{classical} communication from each sender. The protocol exploits the multipartite phase correlations of the GHZ state to encode the local inputs and recover the required global relation at the receiver. In contrast, any protocol achieving a high success probability without preshared entanglement requires polynomial communication from at least one sender. Remarkably, this lower bound continues to hold even when unrestricted \emph{quantum} communication from the senders to the receiver is permitted. Consequently, classical communication assisted by multipartite entanglement is exponentially more powerful than unrestricted quantum communication without preshared entanglement. Our result establishes a multipartite analogue of the communication separation previously established in the bipartite setting~\cite{Gavinsky2006}, thereby identifying multipartite entanglement as a resource for exponential communication savings.

Beyond communication complexity, our results have implications for bounded-storage cryptography with quantum side-information~\cite{Damgard2005}. The multipartite communication protocol developed here naturally gives rise to a seeded randomness extractor involving multiple independent weak sources (i.e., non-uniform but with some guaranteed min-entropy). The communication lower bound of our multipartite task translates directly into a lower bound on the quantum side-information required to compromise the resulting extractor. In particular, our construction shows that preshared entanglement exponentially enhances the power of quantum side-information: compromising the extractor requires polynomial-size unentangled quantum memory, whereas exponentially smaller entangled quantum memory suffices. Consequently, our results establish an exponential separation between entangled and unentangled quantum side-information.

\section{Multiparty communication complexity}

We now formally introduce the multipartite communication tasks underlying these results.

\emph{Multipartite Hidden Matching ($m\mathtt{HM}_n$).} The task $m\mathtt{HM}_n$ involves $m$ spatially separated senders $\{\mathrm{Alice}_r\}_{r=1}^{m}$ and a receiver Bob. Each sender receives an input string: ${\bf x}^r\in\{0,1\}^n$ for $\mathrm{Alice}_r$, where $n$ is even. Bob receives a matching $\mathrm{M}=\{(i_\ell,j_\ell)\,:\,\ell=1,\ldots,n/2\}$ from the set~$\mathscr{M}_{n/2}$ of all perfect matchings on the set $[n]=\{1,\ldots,n\}$. We will assume throughout for simplicity that $n$ is a power of~2. 

The communication model is one-way: each Alice may send a message to Bob, but no communication is allowed among the Alices or from Bob to any of the Alices. Bob must output a triple $\big\langle i_\ell, j_\ell,\bigoplus_{r=1}^{m}\bigl(x^r_{i_\ell}\oplus x^r_{j_\ell}\bigr)\big\rangle$, for some edge $(i_\ell,j_\ell)\in \mathrm{M}$ (see Fig.~\ref{Fig}). Thus, Bob reports one edge of the matching together with the parity, across all senders, of the corresponding endpoint bits. 

For $m=1$, this is the Hidden Matching problem introduced in~\cite{BarYossef2004}. Alice, when allowed to communicate qubits to Bob, encodes her input ${\bf x}=x_1\ldots x_n$ into the $\log n$-qubit state 
\begin{align}\label{hm-Bar}
\ket{\psi_{\bf x}}=U_{\bf x}\left(\tfrac{1}{\sqrt n}\sum_{i=1}^{n}\ket{i}\right),\quad U_{\bf x}:=\sum_{k=1}^{n}(-1)^{x_k}\ket{k}\bra{k},    
\end{align}
and sends it to Bob. Given a perfect matching $\mathrm{M}$, Bob performs the projective measurement $\mathfrak{M}_{\mathrm M}=\{\mathrm P^{i_\ell j_\ell}=\ket{i_\ell}\bra{i_\ell}+\ket{j_\ell}\bra{j_\ell}:(i_\ell,j_\ell)\in\mathrm M\}$, which projects onto a uniformly random edge $(i_\ell,j_\ell)\in\mathrm M$ with resulting state being $\ket{\psi_{i_\ell j_\ell}}=\frac{1}{\sqrt2}\big[(-1)^{x_{i_\ell}}\ket{i_\ell}+(-1)^{x_{j_\ell}}\ket{j_\ell}\big]$. Finally, Bob measures in the basis $\{(\ket{i_\ell}\pm\ket{j_\ell})/\sqrt2\}$, which deterministically reveals the parity $x_{i_\ell}\oplus x_{j_\ell}$. Hence, the Hidden Matching relation $\mathtt{HM}_n$ is solved with perfect success using only $\log n$ qubits of communication.

\emph{Multipartite Boolean Hidden Matching ($m\mathtt{BHM}_n$).} 
Notably, $m\mathtt{HM}_n$ is a relational communication problem: for each input $({\bf x},\mathrm{M})$, there is a valid output for each of the edges of $\mathrm{M}$. We also consider a Boolean decision version of the task, denoted $m\mathtt{BHM}_n$, and defined as follows. A $1/4$-matching $\mathrm{M}$ is a partial matching consisting of $n/4$ disjoint pairs from $[n]$. In addition to such an $\mathrm{M}\in\mathscr{M}_{n/4}$, Bob receives a string ${\bf w}\in\{0,1\}^{n/4}$ satisfying the promise $\mathbb{M}\!\left(\bigoplus_{r=1}^{m}{\bf x}^r\right)\oplus{\bf b}^{n/4}={\bf w}$, where $\mathbb{M}$ is the $n/4\times n$ incidence matrix of $\mathrm{M}$, with one weight-2 row for each edge of $\mathrm{M}$, and ${\bf b}^{n/4}$ is the length-$n/4$ all-$b$ string. Equivalently, the promise ensures either ${\bf z}:=\mathbb{M}\!\left(\bigoplus_{r=1}^{m}{\bf x}^r\right)={\bf w}$ when $b=0$, or ${\bf z}=\overline{\bf w}$ when $b=1$. Bob's goal is to determine the value of the bit $b$, so this problem is now a partial Boolean function. The reason we are defining this with respect to a partial matching $\mathrm{M}$ rather than a perfect matching is that our main lower bound fails in the latter case,  as will become evident in the discussion following Theorem~\ref{th:mainlowerbound} (Appendix~\ref{ApprndixA}).

\begin{figure}[t!]
\centering
\includegraphics[width=1\linewidth]{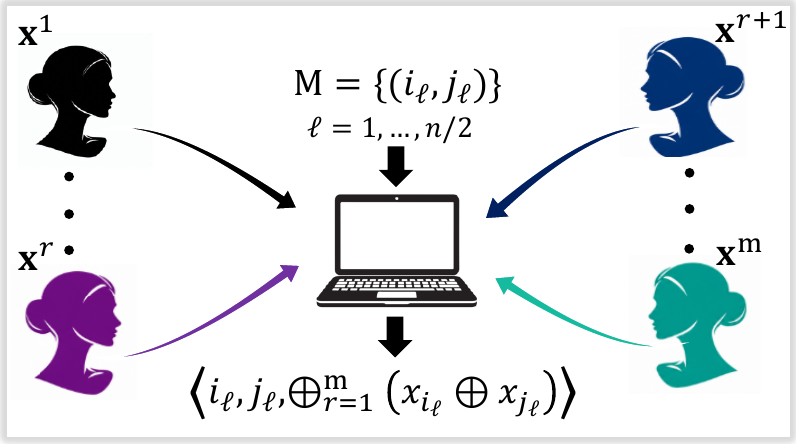}
\caption{{\bf Multipartite Hidden Matching task ($m\mathtt{HM}_n$).} It involves $m$ spatially separated senders (Alices), where the $r$th Alice receives a string ${\bf x}^r\in\{0,1\}^n$. Bob receives a perfect matching $\mathrm{M}\in\mathscr{M}_{n/2}$ on $[n]$. His goal is to output a triple $\big\langle i_\ell,j_\ell,\bigoplus_{r=1}^{m}(x^r_{i_\ell}\oplus x^r_{j_\ell})\big\rangle$ for some edge $(i_\ell,j_\ell)\in\mathrm{M}$. The Boolean variant $m\mathtt{BHM}_n$ is a related decision problem defined in the text.}
\label{Fig}
\vspace{-.25cm}
\end{figure}

Within the Simultaneous Message Passing (SMP) framework \cite{Babai1997}, a natural hierarchy of one-way communication models emerges, distinguished by the communication resources and shared correlations available to the parties:

\begin{itemize}[itemsep=-.05cm, topsep=2pt, leftmargin=.3cm]

\item \textit{Classical communication without shared resources:} The Alices send classical messages to Bob, and no shared randomness or entanglement is available. The corresponding communication complexity (i.e., worst-case total number of bits sent, minimized over all protocols solving the task) for perfect, error-free protocols is denoted by $C^{\parallel}$.

\item \textit{Classical communication with global shared randomness:} All parties have access to a common random string. The corresponding complexity is denoted by $C^{\parallel,\mathrm{GSR}}$ for error-free protocols.

\item \textit{Classical communication assisted by bipartite entanglement:} Each Alice shares independent bipartite entanglement with Bob, yielding complexity $C^{\parallel,\{\mathrm{Ent}_{A_iB}\}_{i=1}^{m}}$. If, additionally, bipartite entanglement is shared between every pair of Alices, the complexity is denoted by
$C^{\parallel,\{\mathrm{Ent}_{A_iB},\mathrm{Ent}_{A_iA_j}\}_{i,j=1}^{m}}$ for error-free protocols.

\item \textit{Classical communication assisted by genuine multipartite entanglement:} All parties share a genuine multipartite entangled state. The corresponding complexity is denoted $C^{\parallel,\mathrm{GEnt}}$ for perfect protocols.
\end{itemize}
``Error-free'' above means that the protocol must produce a correct output with probability~1, for each possible input.
We will also consider communication protocol that are allowed some error probability, at most $1/3$ on every possible input, and we will add a subscript $1/3$ below the $C$ to indicate that.

Replacing the classical channels from the Alices to Bob by quantum channels gives rise to the corresponding quantum communication complexities $Q^{\parallel}$, $Q^{\parallel,\mathrm{GSR}}$, $Q^{\parallel,\{\mathrm{Ent}_{A_iB}\}_{i=1}^{m}}$, $Q^{\parallel,\{\mathrm{Ent}_{A_iB},\mathrm{Ent}_{A_iA_j}\}_{i,j=1}^{m}}$, and $Q^{\parallel,\mathrm{GEnt}}$, with a subscript $1/3$ for their bounded-error versions. 

We begin by presenting efficient classical communication protocol for $m\mathtt{HM}_n$ and $m\mathtt{BHM}_n$ when a multipartite GHZ state is shared among all the parties.

\begin{theorem}\label{theo1}
For every $m\ge1$, $C^{\parallel,\mathrm{GEnt}}(m\mathtt{HM}_n)\le \log_2 n$ and $C^{\parallel,\mathrm{GEnt}}_{1/3}(m\mathtt{BHM}_n)\leq \log_2 n$.
\end{theorem}

\begin{proof}
The $m$ senders $\{\mathrm{Alice}_r\}_{r=1}^{m}$ and Bob share the $(m+1)$-partite $n$-dimensional GHZ state $\ket{\mathcal G}_{m+1}=\frac{1}{\sqrt n}\sum_{k=1}^{n}(\bigotimes_{r=1}^{m}\ket{k}_{A_r})\ket{k}_B$, where each party holds $\log_2 n$ qubits. Upon receiving ${\bf x}^r=x^r_1\ldots x^r_n$, Alice$_r$ applies the unitary $U^{(r)}_{\bf x}=\sum_{k=1}^{n}(-1)^{x_k^r}\ket{k}\bra{k}$ on her share of the GHZ state, thereby transforming the shared state into $\ket{\mathcal G({\bf x}^1,\ldots,{\bf x}^m)}_{m+1}=\frac{1}{\sqrt n}\sum_{k=1}^{n}(-1)^{p_k}(\bigotimes_{r=1}^{m}\ket{k}_{A_r})\ket{k}_B$, where $p_k=\bigoplus_{r=1}^{m}x_k^r$. Each Alice then measures her subsystem in the Fourier basis $\mathbb{FB}\equiv\big\{\ket{f_a}:=\frac{1}{\sqrt n}\sum_{k=1}^{n}\omega^{ka}\ket{k}\big\}_{a=1}^{n}$, $\omega:=e^{2\pi i/n}$; and communicates the $\log_2 n$-bit outcome ($o_r$ for the $r$th Alice) to Bob. Conditioned on the outcome tuple $\vec{o}=(o_1,\ldots,o_{m})$, Bob's marginal state is 
\begin{align}
\ket{\psi(\vec{o})}_B=\frac{1}{\sqrt n}\sum_{k=1}^{n}(-1)^{p_k}\omega^{-k\eta}\ket{k}_B,~\eta:=\sum_{r=1}^{m}o_r.
\end{align}
Applying the unitary $U_B:=\sum_{k=1}^{n}\omega^{k\eta}\ket{k}\bra{k}$, Bob's state becomes
\begin{align}\label{hm-multi}
\ket{\psi}_B
=\frac{1}{\sqrt n}\sum_{k=1}^{n}
(-1)^{p_k}\ket{k}_B,
\end{align}
which is exactly the Hidden Matching encoding of the parity string ${\bf p}=(p_1,\ldots,p_n)$ [see Eq.~\eqref{hm-Bar}]. Bob then follows the same measurement procedure as in the bipartite protocol, obtaining $p_{i_\ell}\oplus p_{j_\ell}=\bigoplus_{r=1}^{m}(x_{i_\ell}^r\oplus x_{j_\ell}^r)$.

For the Boolean version $m\mathtt{BHM}_n$, Bob completes the $1/4$-matching $\mathrm{M}$ to a perfect matching (adding $n/4$ edges in some arbitrary way), and measures the state \eqref{hm-multi} as before. With probability $1/2$ the obtained edge $(i_\ell,j_\ell)$ is in $\mathrm{M}$. Bob computes $z_\ell=\bigoplus_{r=1}^{m}\bigl(x^r_{i_\ell}\oplus x^r_{j_\ell}\bigr)$ from the measurement outcome and compares it with $w_\ell$. By the promise, $z_\ell=w_\ell$ iff $b=0$, and $z_\ell\neq w_\ell$ iff $b=1$. Hence, if indeed $(i_\ell,j_\ell)\in\mathrm{M}$, then Bob recovers $b$ with certainty. 
If the obtained $(i_\ell,j_\ell)$ was not in $\mathrm{M}$ then Bob outputs a random coin flip. His overall error probability is now $1/4<1/3$.
\end{proof}

\emph{Remark:} A notable distinction between the bipartite and multipartite settings concerns information leakage. In the bipartite task, let's say with uniform input distribution, any successful protocol necessarily reveals partial information about Alice's input to Bob, namely the parity of the two input bits associated with the reported edge of the matching. By contrast, in multipartite protocol of Theorem~\ref{theo1}, Bob learns only the global parity $\bigoplus_{r=1}^{m}\bigl(x^r_{i_\ell}\oplus x^r_{j_\ell}\bigr)$, while gaining no information about the input string of any individual sender, or even about the joint input strings of any $m-1$ of the senders. Thus, the protocol achieves perfect success while remaining completely oblivious to each sender's local input. 

\medskip

We now turn to communication lower bounds in the absence of preshared entanglement. By fixing the inputs for $m-1$ senders to the all-0 string $0^n$, we immediately get a lower bound for the $m$-partite task from the bipartite task:

\begin{lemma}\label{lemma1}
For every $m\ge 1$, $C^{\parallel,\mathrm{GSR}}(m\mathtt{HM}_n)\ge C^{\parallel,\mathrm{GSR}}(\mathtt{HM_n})$ and $C^{\parallel,\mathrm{GSR}}(m\mathtt{BHM}_n)\ge C^{\parallel,\mathrm{SR}}(\mathtt{BHM_n})$, and similarly for the bounded-error versions. 
\end{lemma}

Using the known lower bounds for the bipartite (i.e., $m=1$) hidden matching problems \cite{BarYossef2004,Buhrman2012}, Lemma~\ref{lemma1} implies $C_{1/3}^{\parallel,\mathrm{GSR}}(m\mathtt{HM}_n)\ge\Omega(\sqrt n)$ and $C^{\parallel,\mathrm{GSR}}_{1/3}(m\mathtt{BHM}_n)\ge\Omega(\sqrt{n})$~\cite{Gavinsky2007} for bounded-error protocols. 

We next strengthen these lower bounds by allowing unrestricted \emph{quantum} communication while still forbidding preshared entanglement, namely in the communication model $Q^{\parallel,\mathrm{GSR}}$.     

\begin{theorem}\label{theo2}
For all $m\ge 2$ the bounded-error communication complexity  $Q_{1/3}^{\parallel,\mathrm{GSR}}(m\mathtt{BHM}_n)=\Omega(\sqrt{n})$, in fact at least one sender has to send $\Omega(\sqrt{n})$ qubits.
\end{theorem}

Due to space constraints, we defer the proof to Appendix~\ref{ApprndixA}. For comparison, the bipartite Hidden Matching problem admits a one-way quantum protocol requiring only $O(\log n)$ qubits of communication from Alice to Bob~\cite{BarYossef2004}. In contrast, Theorem~\ref{theo2} shows that if we have two or more unentangled Alices, then at least one of them needs to send $\Omega(\sqrt n)$ qubits. Combined with Theorem~\ref{theo1}, this establishes an exponential separation between entanglement-assisted classical communication and unassisted quantum communication in our multiparty setting. 

\emph{Shared entanglement (between sender and receiver).} Consider the communication model $Q^{\parallel,\{\mathrm{Ent}_{A_iB}\}_{i=1}^{m}}$, where each Alice shares $r$ EPR pairs with Bob. Let the shared state be $\Phi_{AB}^{\otimes r}$. Upon receiving input ${\bf x}$, Alice can apply a CPTP map $\mathcal{G}_{\bf x} : A \rightarrow \mathcal{R} \otimes A'$ where $\mathcal{R}$ is a $q$-qubit message register and $A'$ is an auxiliary system retained by Alice. The global state becomes  $\rho^{({\bf x})}_{\mathcal{R}A'B}=(\mathcal{G}_{\bf x} \otimes I_{B})(\Phi^{\otimes r}_{AB})$. Alice then sends the register $\mathcal{R}$ to Bob. Consequently, Bob receives the system $\mathcal{R}B$ in state  $\rho^{({\bf x})}_{\mathcal{R}B}=\operatorname{Tr}_{A'}\!\left[\rho^{(x)}_{\mathcal{R}A'B}\right]$, which acts on a $(q+r)$-qubit Hilbert space. Notably this protocol can be simulated without any prior entanglement. Indeed, Alice can locally prepare the mixed state $\rho^{({\bf x})}_{\mathcal{R}B}$ and transmit the entire $(q+r)$-qubit system to Bob, who thereby receives exactly the same state as in the original protocol. Hence any protocol using $q$ qubits of communication and $r$ shared EPR pairs induces a protocol without preshared entanglement using $q+r$ qubits of communication. Applying Theorem~\ref{theo2}, we obtain $q+r = \Omega(\sqrt{n})$. If the amount of prior entanglement is limited, i.e., $r \leq c\sqrt{n}$ for sufficiently small constant $c$, then the bound implies $q = \Omega(\sqrt{n})$. However, if the shared entanglement is large ($r \gg \sqrt{n}$), then the condition $q + r \ge \Omega(\sqrt{n})$ does not impose a meaningful constraint on $q$, since the bound may already be satisfied by $r$ alone.

\emph{Shared entanglement (among senders).} The above lower-bound argument does not extend to the communication model $Q^{\parallel,\mathrm{Ent}_{A_1\ldots A_m}}$, in which the senders share preshared entanglement. In this setting, the quantum messages received by Bob are no longer tensor-product states, and consequently the simulation argument used above no longer applies. Indeed, suppose the senders $\{\mathrm{Alice}_r\}_{r=1}^{m}$ initially share the $m$-partite GHZ state $\ket{\mathcal G}_{m}=\tfrac{1}{\sqrt n}\sum_{k=1}^{n}\bigotimes_{r=1}^{m}\ket{k}_{A_r}$, where each party holds $\log_2 n$ qubits (in contrast to the previous, Bob does not share part of this GHZ state). Upon receiving their respective inputs, Alice$_r$ applies the phase encoding $U_{{\bf x}^r}$ to her subsystem and sends it to Bob using only $O(\log n)$ qubits of communication. Bob then holds the state
\begin{align}
\ket{\psi}_m=\frac{1}{\sqrt n}\sum_{k=1}^{n}(-1)^{p_k}\ket{k}_A,\quad \ket{k}_A:=\bigotimes_{r=1}^{m}\ket{k}_{A_r},   
\end{align}
which takes a form similar to the state in Eq.~\eqref{hm-Bar}. Therefore, Bob can do measurement based on his input $\mathrm{M}$ and solve the relational Hidden Matching problem.

\section{Application to Bounded-Storage Cryptography} 

The Hidden Matching problem induces a strong seeded randomness extractor, as observed in~\cite{Gavinsky2007}. Randomness extractors are important tools in theoretical computer science~\cite{Shaltiel2004} and cryptography~\cite{Bennett1988,Impagliazzo1989,Konig2005}. Before discussing the explicit construction, we recall two definitions.

\begin{definition}[Weak random source]
A random variable ${\bf x}$ over the set $\{0,1\}^n$ is called a \emph{$k$-min-entropy source}, or simply a \emph{weak random source}, if its min-entropy satisfies $H_{\infty}({\bf x})\ge k$.
Here the \emph{min-entropy} of ${\bf x }$ is defined as $H_{\infty}({\bf x})=-\log_2\left(\max_{{\bf x}\in\{0,1\}^n}\Pr[{\bf x}=\mathtt{x}]\right)$.
Thus, if $H_{\infty}({\bf x})=k$, then no individual outcome of the random variable occurs with probability greater than $2^{-k}$.
\end{definition}

\begin{definition}[Strong Seeded Randomness Extractor]
Let $n,\ell,d,k\in\mathbb{N}$ and let $\varepsilon>0$. A function
$\mathrm{Ext}:\{0,1\}^{n}\times\{0,1\}^{d}\rightarrow\{0,1\}^{\ell}$
is called a strong \emph{$(k,\varepsilon)$-seeded randomness extractor} if, for every random variable ${\bf x}$ over $\{0,1\}^{n}$ satisfying $
H_{\infty}({\bf x})\ge k$, and an independent random variable
${\bf y}$ that is uniform over a finite set $S$ (of possible seeds),
the distribution of $({\bf y},
\mathrm{Ext}({\bf x},{\bf y}))$
is $\varepsilon$-close to the uniform distribution on $S\times \{0,1\}^{\ell}$, equivalently
\[
\E_{{\bf y}}
\norm{
\E_{{\bf x}}\left[ \mathrm{Ext}({\bf x},{\bf y})- U_\ell
\right]
}_{tvd}\leq\varepsilon.
\]
Here total variation distance (tvd) between probability distributions $P$ and~$Q$ is defined as $\norm{P-Q}_{\mathrm{tvd}}=\frac12\sum_{r\in\{0,1\}^{\ell}}|P(r)-Q(r)|$.
\end{definition} 

We first explain the result of~\cite{Gavinsky2007}.
Let ${\bf x}\in\{0,1\}^{n}$ be a weak random source and let $\mathrm{M}\in\mathscr{M}_{n/4}$ denote a uniformly random $1/4$-matching. The Hidden Matching extractor $\mathrm{Ext}:\{0,1\}^{n}\times\mathscr{M}_{n/4}\rightarrow\{0,1\}^{n/4}$ 
defined as $\mathbb{M}({\bf x})$, where $\mathbb{M}$ is an $n/4\times n$ matrix over $\mathbb{F}_2$, with one weight-2 row for each edge of the $1/4$-matching $\mathrm{M}$.
It produces an $n/4$-bit output ${\bf z}=\mathrm{Ext}({\bf x},\mathrm{M})$ that remains statistically close to uniform even when the seed $\mathrm{M}$ is public. The security of the resulting extractor against classical memory-bounded adversaries is very similar to the communication lower bound of~\cite{Gavinsky2007} for Boolean Hidden Matching, relying on the Fourier coefficients inequality of Kahn, Kalai, and Linial (KKL)~\cite{Kahn1988}. If the source has min-entropy $H_{\infty}({\bf x})\ge n-\gamma\varepsilon\sqrt{n}$, then ${\bf z}$ is $\varepsilon$-close to uniform, implying security against classical bounded-storage adversaries with memory $s=\gamma\varepsilon\sqrt{n}$~\cite{Gavinsky2007} (note that we fixed the $\alpha$ of~\cite{Gavinsky2007} to be 1/4 in this paper). In contrast, a \emph{quantum} adversary storing the $\log(n)$-qubit state 
$\ket{\psi_{\bf x}}=\frac{1}{\sqrt n}\sum_{i=1}^{n}(-1)^{x_i}\ket{i}$
can recover one of the Hidden Matching parities (which is one of the bits of ${\bf z}$) once the seed~$\mathrm{M}$ is revealed, as explained after Eq.~\eqref{hm-Bar}. It should be noted that this extractor $\mathrm{Ext}$ is not intended to optimize extraction parameters, as it incurs a significant loss of source min-entropy and requires a long random seed~$\mathrm{M}$ (much more efficient quantum-proof extractors are known, e.g.~\cite{De2012Trevisan,BenAroyaTaShma2012}). Rather, its purpose is to establish a fundamental separation: a function that is a strong extractor against classical storage can completely fail in the presence of small quantum storage. 

\begin{figure}[t!]
\centering
\includegraphics[width=1.0\linewidth]{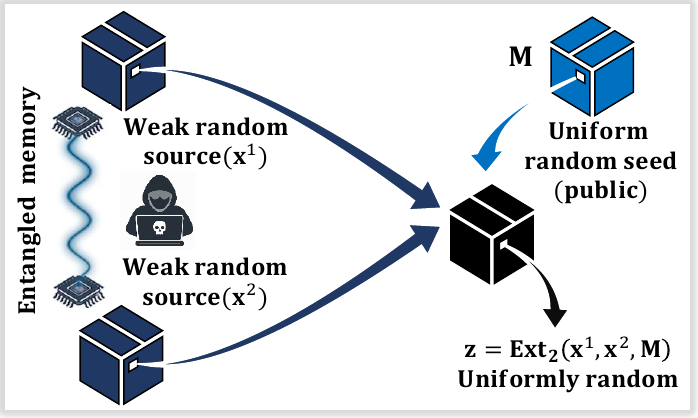}
\caption{\textbf{Multipartite Hidden Matching extractor.} Two independent weak random sources ${\bf x}^{1}$ and ${\bf x}^{2}$, together with a uniformly random public seed $\mathrm{M}$, are processed by the multipartite Hidden Matching extractor to produce an extracted key $ {\bf z}=\mathrm{Ext}_2({\bf x}^{1},{\bf x}^{2},\mathrm{M})$, that is statistically close to uniform. Adversaries with unentangled memory require polynomial-size quantum side-information to attack the protocol. In contrast, pre-shared entanglement enables an adversary to compromise the extractor using only $O(\log n)$ qubits of storage per source, establishing an exponential separation between entangled and unentangled quantum memory.}
\label{fig2}
\end{figure}

We generalize this one-source construction to multiple independent weak sources by introducing a two-source Hidden Matching extractor: $\mathrm{Ext}_2:\{0,1\}^{n}\times\{0,1\}^{n}\times\mathscr{M}_{n/4}\rightarrow\{0,1\}^{n/4}$ that combines the Hidden Matching parities of ${\bf x}^1$ and ${\bf x}^2$ as $\mathrm{Ext}_2({\bf x}^1,{\bf x}^2,\mathrm{M})=\mathbb{M}({\bf x}^1\oplus {\bf x}^2)$. This maps two independent $n$-bit sources (${\bf x}^1$ and ${\bf x}^2$, each with a lower bound on their min-entropy) and a uniformly random seed (the $1/4$-matching~$\mathrm{M}$) to an $n/4$-bit string that is intended to be close to uniform, see Fig.~\ref{fig2}. Although multi-source randomness extractors have been studied extensively for decades \cite{Barak2006,Barak2012,Chattopadhyay2019}, comparatively little is known about their security in the presence of entangled quantum side information \cite{Kasher2012}. 

In the case of extractors with classical side-information about the source(s), that side-information can be folded into the min-entropy condition. However, to deal with quantum side-information about the source(s) we need the following definition:

\begin{definition}[Strong Seeded Two-Source Randomness Extractor with Quantum side-information]
Let $n,\ell,d,k\in\mathbb{N}$ and let $\varepsilon>0$. A function
$\mathrm{Ext}:\{0,1\}^{n}\times\{0,1\}^{n}\times\{0,1\}^{d}\rightarrow\{0,1\}^{\ell}$
is called a strong \emph{$(k,\varepsilon)$-seeded two-source randomness extractor} against quantum side-information map $({\bf x}^1,{\bf x}^2)\mapsto\rho({\bf x}^1,{\bf x}^2)$ if, for every pair of independent random variables ${\bf x}^1,{\bf x}^2$ over $\{0,1\}^{n}$ satisfying $
H_{\infty}({\bf x}^1)\ge k,H_{\infty}({\bf x}^2)\ge k$, and an independent random variable
${\bf y}$ that is uniform over a finite set $S$ (of possible seeds),
the distribution of $
({\bf y},\mathrm{Ext}({\bf x}^1,{\bf x}^2,{\bf y}))$
is $\varepsilon$-close to the uniform distribution on $S\times \{0,1\}^{\ell}$ even given $\rho{(\bf x}^1,{\bf x}^2)$, equivalently
\[
\scriptstyle
\E_{{\bf y}}\norm{
\E_{{\bf x}^1,{\bf x}^2}\left[
\rho({\bf x}^1,{\bf x}^2)\otimes \mathrm{Ext}({\bf x}^1,{\bf x}^2,{\bf y})-\rho({\bf x}^1,{\bf x}^2)\otimes U_\ell
\right]
}_{tvd}\leq\varepsilon.
\]
Here $\norm{\cdot}_{tvd}$ is the quantum generalization of total variation distance (=sum of singular values divided by 2).
\end{definition} 

We show that this multipartite construction reveals a qualitatively new distinction between two unentangled and entangled \emph{quantum} side-information:

\begin{theorem}\label{theo3}
{\rm [Entangled vs Unentangled side-information]} 
Fix $\varepsilon\in(0,1/4)$.
The two-source Hidden Matching extractor~$\mathrm{Ext}_2$ exhibits an exponential separation between unentangled and entangled side-information:
\begin{itemize}[itemsep=-.05cm, topsep=2pt, leftmargin=.6cm]
\item[(1)] there exists a positive constant $K$ such that if the side-information consists of two unentangled  states of $q=K\sqrt{\varepsilon n}$-qubits each, one depending on ${\bf x}^1$ and one depending on ${\bf x}^2$, then $\mathrm{Ext}_2$ is a strong $(n-K\sqrt{\varepsilon n},\varepsilon)$-seeded two-source randomness extractor;
\item[(2)] if the side-information may consist of two entangled states of $q=\log n$ qubits, one depending on ${\bf x}^1$ and one depending on ${\bf x}^2$, then $\mathrm{Ext}_2$ is not a strong $(n,\varepsilon)$-seeded two-source randomness extractor.
\end{itemize}
\end{theorem}
Part~(2) is easy: the entangled side-information can be 
\[
\left(U_{{\bf x}^1}\otimes
U_{{\bf x}^2}\right)\frac{1}{\sqrt{n}}\sum_{i=1}^n\ket{i}\ket{i}=\frac{1}{\sqrt{n}}\sum_{i=1}^n(-1)^{x^1_i\oplus x^2_i}\ket{i}\ket{i},
\]
and one bit of $\mathrm{Ext}_2({\bf x}^1,{\bf x}^2,\mathrm{M})$ can be learned (with probability 1/2, knowing when we succeed) by doing the usual
$\mathrm{M}$-dependent measurement.
This one bit gives
\[\scriptstyle
\E_{{\bf y}}\norm{
\E_{{\bf x}^1,{\bf x}^2}\left[
\rho({\bf x}^1,{\bf x}^2)\otimes \mathrm{Ext}({\bf x}^1,{\bf x}^2,{\bf y})-\rho({\bf x}^1,{\bf x}^2)\otimes U_\ell
\right]
}_{tvd}\geq 1/4.
\]
Our analysis for unentangled quantum adversaries to establish part~(1) is substantially more delicate than the security analysis against classical adversaries in~\cite{Gavinsky2007}, because the adversary now possesses quantum side-information about two independent weak sources rather than classical side-information. It relies on the hypercontractive inequality for matrix-valued functions developed by Ben-Aroya, Regev, and de Wolf~\cite{BenAroya2008}, which generalizes the KKL inequality to the quantum setting. A complete proof is provided in Appendix~\ref{ApprndixA}. The multipartite Hidden Matching extractor therefore establishes an exponential separation between entangled and unentangled quantum side-information, identifying  entanglement as an important resource in multi-source bounded-storage cryptography. As in the one-source Hidden Matching extractor of~\cite{Gavinsky2007}, our two-source construction is intended as a proof-of-principle rather than an optimized extractor. Its primary significance is to demonstrate that security against unentangled quantum side-information does not necessarily extend to adversaries possessing even a modest amount of preshared entanglement.

\section{Discussion} 
While the communication advantages of bipartite quantum resources have been extensively investigated, comparatively little has been known about the computational power of multipartite entanglement. We show that classical communication assisted multipartite entanglement can outperform unassisted quantum communication exponentially in a distributed computational task, establishing that preshared GHZ entanglement is not merely a substitute for quantum communication but can be a strictly stronger information-processing resource. 

Beyond communication complexity, our work reveals a novel connection between multipartite communication protocols and bounded-storage cryptography. By extending the Hidden Matching extractor to multiple independent weak random sources, we establish an exponential separation between entangled and unentangled quantum side-information. Specifically, compromising the extractor requires polynomial-size unentangled quantum memory, whereas preshared entanglement reduces the required quantum memory exponentially. 

More broadly, our work suggests that multipartite entanglement should be regarded as a computational resource whose power extends beyond the generation of nonlocal correlations. An important open question is to characterize the class of distributed computational tasks that admit exponential entanglement-assisted communication advantages and to determine which families of multipartite entangled states realize such advantages. Another natural direction is to investigate whether similar separations persist in interactive communication models, noisy communication settings, or networked quantum architectures. In particular, for $m\geq 3$, can we strengthen our polynomial lower bound on the required quantum communication when the Alices share pairwise entanglement, such as EPR pairs? (for $m=2$, pairwise entanglement between the two Alices is just a GHZ state, which allows an efficient protocol).

Our cryptographic results likewise open several directions for future research. The multipartite Hidden Matching extractor presented here is intended as a proof-of-principle construction rather than an optimized extractor, raising the question of whether similar entanglement-induced separations can be achieved with randomness extractors possessing near-optimal parameters. More generally, it will be interesting to investigate similar qualitative distinction between entangled and unentangled quantum side-information for other cryptographic primitives, including privacy amplification, secret-key agreement, and secure multi-party computation. We expect that these connections between multipartite entanglement, communication complexity, and quantum cryptography will stimulate further developments in distributed quantum information processing and secure quantum networks.

\medskip

\noindent{\bf Acknowledgement}: MB acknowledges the financial support through the National Quantum Mission (NQM) of the Department of Science and Technology, Government of India. RdW is partially supported by the Dutch Research Council (NWO) through Gravitation-grant Quantum Software Consortium, 024.003.037.

\smallskip

\noindent{\bf AI statement}: All the ideas and writing are human. ChatGPT 5.6 pro was used for a sanity-check on the near-final manuscript, revealing some small errors in the writing that we corrected.


%

\onecolumngrid

\appendix
\section{Proof of Theorem~\ref{theo2} \& Theorem~\ref{theo3}}\label{ApprndixA}

\noindent We begin by proving a more general theorem and then will argue that Theorems~\ref{theo2} and~\ref{theo3} follow as special cases. Consider the function $Z:\01^n\times\01^n\times\mathscr{M}_{n/4}\to\01^{n/4}$ defined by  $Z({\bf x}^1,{\bf x}^2,\mathrm{M})=\mathbb{M}({\bf x}^1\oplus {\bf x}^2)$, where $\mathbb{M}$ is an $n/4\times n$ matrix over $\mathbb{F}_2$, with one weight-2 row for each edge of the $1/4$-matching $\mathrm{M}$. This maps two $n$-bit sources (${\bf x}^1$ and ${\bf x}^2$) and a seed (the $1/4$-matching~$\mathrm{M}$) to an $n/4$-bit string that is intended to be close to uniformly random. We sometimes write $Z_\mathrm{M}({\bf x}^1,{\bf x^2)}$ if $\mathrm{M}$ is fixed while ${\bf x}^1,{\bf x}^2$ are random. The analysis below roughly follows the proof of \cite{Gavinsky2007}, with the hypercontractive inequality for matrix-valued functions of~\cite{BenAroya2008} replacing the standard hypercontractive inequality in order to deal with quantum storage instead of classical, but getting all the details to fit together is quite nontrivial.

\begin{theorem}\label{th:mainlowerbound}
Let $0\leq c\leq q\leq n$ be integers. Let $A,B\subseteq\01^n$ each have size $\geq 2^{n-c}$. Let $f_1({\bf x}^1)$ assign a $q$-qubit density matrix for each ${\bf x}^1\in A$, and similarly let $f_2({\bf x}^2)$ assign a $q$-qubit density matrix for each ${\bf x}^2\in B$.
Let ${\bf x}^1,{\bf x}^2$ be uniformly distributed over $A,B$ respectively, and $\mathrm{M}$ be uniform over the set of all $1/4$-matchings.
Let
\[
\rho_\mathrm{M} = \E_{{\bf x}^1\in A,{\bf x}^2\in B}\left[f_1({\bf x}^1)\otimes f_2({\bf x}^2)\otimes \ketbra{Z_\mathrm{M}({\bf x}^1,{\bf x}^2)}{Z_\mathrm{M}({\bf x}^1,{\bf x}^2)}\right]
\]
be the state of a party who has the quantum states $f_1({\bf x}^1),f_2({\bf x}^2)$ and the matching $\mathrm{M}$ (but not ${\bf x}^1,{\bf x}^2$), alongside the value of the function $Z_\mathrm{M}$ (unknown to that party).
Let 
$
\Delta_{\mathrm{M}}=\E_{{\bf x}^1\in A,{\bf x}^2\in B}\left[f_1({\bf x}^1)\otimes f_2({\bf x}^2)\otimes \left(\ketbra{Z_\mathrm{M}({\bf x}^1,{\bf x}^2)}{Z_\mathrm{M}({\bf x}^1,{\bf x}^2)}-\frac{I}{2^{n/4}}\right)\right]$
be the difference between $\rho_{\mathrm{M}}$, and the state where the last register is maximally mixed (i.e., uniform over $\01^{n/4}$). Then
\[
\E_{\mathrm{M}}[\norm{\Delta_{\mathrm{M}}}_{tvd}]\leq O(q^2/n).
\]
\end{theorem}

Note that this theorem is false if $\mathrm{M}$ is a \emph{perfect} matching even if $c=0$ and $q=1$: $f_1({\bf x}^1)$ and $f_2({\bf x}^2)$ could just be the parity of their respective inputs. The parity of those two bits equals the parity of the $n/2$ bits of $Z_\mathrm{M}({\bf x}^1,{\bf x}^2)$ in the case where $\mathrm{M}$ is a perfect matching, so you would already break the uniformity property. 

\begin{proof}
We may assume without loss of generality that $q$ is positive and even, and at most $\sqrt{n}$ (the conclusion is trivially true if $q>\sqrt{n}$).
Define the functions $f_1,f_2$ to be 0 outside of $A,B$ respectively.
Fix a $1/4$-matching $\mathrm{M}$. For a function $f:\01^n\to{\cal M}_d$,  where ${\cal M}_d$ is the set of $d$-dimensional complex matrices, define its Fourier transform $\widehat{f}:\01^n\to{\cal M}_d$ as 
\[
\widehat{f}({\bf s})=\frac{1}{2^n}\sum_{{\bf x}\in\01^n}f({\bf x})(-1)^{{\bf x}\cdot s}. 
\]
Note that the Fourier coefficient $\widehat{f}({\bf s})$ is a $d$-dimensional matrix here, not a number as is more common in Fourier analysis. We have the Fourier decomposition $f({\bf x})=\sum_{\bf s} \widehat{f}({\bf s})(-1)^{{\bf x}\cdot {\bf s}}$ for all ${\bf x}\in\01^n$.

We will also use a slightly different decomposition for diagonal matrices to analyze the output of the function~$Z_\mathrm {M}$. Define $V:\01^{n/4}\to{\cal M}_{2^{n/4}}$ as 
\[
V({\bf u})=\frac{1}{2^{n/4}}\sum_{{\bf w}\in\01^{n/4}}(-1)^{{\bf u}\cdot {\bf w}}\ketbra{{\bf w}}{\bf {w}}.
\]
Note that $V(0^{n/4})$ is the maximally mixed $n/4$-qubit state $I/2^{n/4}$, and for every ${\bf z}\in\01^{n/4}$ we have the decomposition 
\[
\ketbra{{\bf z}}{{\bf z}}=\sum_{\bf u} V({\bf u})(-1)^{{\bf u}\cdot {\bf z}}.
\]
We now use the decompositions of the functions $f_1$, $f_2$, and $\ketbra{Z_{\mathrm{M}}({\bf x}^1,{\bf x}^2)}{Z_{\mathrm{M}}({\bf x}^1,{\bf x}^2)}$ to analyze $\rho_\mathrm{M}$:
\begin{align*}
\rho_\mathrm{M} & = \E_{{\bf x}^1\in A,{\bf x}^2\in B}\left[f_1({\bf x}^1)\otimes f_2({\bf x}^2)\otimes \ketbra{Z_\mathrm{M}({\bf x}^1,{\bf x}^2)}{Z_\mathrm{M}({\bf x}^1,{\bf x}^2)}\right]\\
& = \frac{2^{2n}}{|A|\cdot|B|}\E_{{\bf x}^1\in\01^n ,{\bf x}^2\in \01^n}\left[f_1({\bf x}^1)\otimes f_2({\bf x}^2)\otimes\ketbra{Z_\mathrm{M}({\bf x}^1,{\bf x}^2)}{Z_\mathrm{M}({\bf x}^1,{\bf x}^2)}\right]\\
& = \frac{2^{2n}}{|A|\cdot|B|}\E_{{\bf x}^1\in\01^n ,{\bf x}^2\in \01^n}\left[\sum_{{\bf s}\in\01^n}\widehat{f_1}({\bf s})(-1)^{{\bf s}\cdot{\bf x}^1}\otimes \sum_{{\bf t}\in\01^n}\widehat{f_2}({\bf t})(-1)^{{\bf t}\cdot {\bf x}^2}
\otimes\sum_{{\bf u}\in\01^{n/4}}V({\bf u})(-1)^{u\cdot \mathbb{M}({\bf x}^1\oplus {\bf x}^2)}\right]\\
& = \frac{2^{2n}}{|A|\cdot|B|}\sum_{{\bf s,t,u}}\widehat{f_1}({\bf s})\otimes \widehat{f_2}({\bf t})\otimes V(\bf{u})
\cdot \E_{{\bf x}^1\in\01^n}\left[(-1)^{{\bf x}^1\cdot (\mathbb{M}^T {\bf u}\oplus {\bf s})}\right]
\cdot \E_{x^2\in \01^n}\left[(-1)^{x^2\cdot (\mathbb{M}^T u\oplus t)}\right]\\
& = \frac{2^{2n}}{|A|\cdot|B|}\sum_{{\bf u}\in\01^{n/4}}\widehat{f_1}(\mathbb{M}^T {\bf u})\otimes \widehat{f_2}(\mathbb{M}^T {\bf u})\otimes V({\bf u}),
\end{align*}
where the last equality is because the $\E_{{\bf x}^1}$ acts like a delta-function on $\mathbb{M}^T{\bf u}$ and ${\bf s}$, and  the $\E_{{\bf x}^2}$ acts like a delta-function on $\mathbb{M}^T{\bf u}$ and ${\bf t}$.

The difference $\Delta_\mathrm{M}$ between the above~$\rho_\mathrm{M}$ and the state where the last register is the maximally mixed state is
\begin{align*}
\Delta_\mathrm{M} = \frac{2^{2n}}{|A|\cdot|B|}\sum_{{\bf u}\in\01^{n/4}\backslash\{0^{n/4}\}}\widehat{f_1}(\mathbb{M}^T {\bf u})\otimes \widehat{f_2}(\mathbb{M}^T {\bf u})\otimes V({\bf u}),
\end{align*}
because $V(0^{n/4})$ equals the maximally mixed state $I/2^{n/4}$ so ${\bf u}=0^{n/4}$ drops out of the sum.

For $p\geq 1$ define the normalized Schatten $p$-norm of a $d$-dimensional matrix $W$ as $\norm{W}_p=(\frac{1}{d}\Tr(|W|^p))^{1/p}$.
One can show that these norms are monotone nondecreasing in~$p$. Note that if $W$ is an $n$-qubit density matrix, then $\norm{W}_p^p\leq 2^{-n}$ for every $p$.

Using triangle inequality and the fact that $\norm{V({\bf u})}_1=2^{-n/4}$ for all ${\bf u}\in\01^{n/4}$,
we bound
\begin{align*}
\norm{\sum_{{\bf u}\in\01^{n/4}\backslash\{0^{n/4}\}}\widehat{f_1}(\mathbb{M}^T {\bf u})\otimes \widehat{f_2}(\mathbb{M}^T {\bf u})\otimes V({\bf u})}_1
 \leq 2^{-n/4}\sum_{k=1}^{n} \sum_{{\bf u}\in\01^{n/4}:|\mathbb{M}^T{\bf u}|=k}\norm{\widehat{f_1}(\mathbb{M}^T u)}_1\cdot\norm{\widehat{f_2}(\mathbb{M}^T u)}_1.
\end{align*}
This analysis of $\Delta_\mathrm{M}$ so far is for a fixed $1/4$-matching $\mathrm{M}$. We now want to show that its norm is small when averaged uniformly over all $1/4$-matchings~$\mathrm{M}$.
We have the following probabilistic lemma  from~\cite[Lemma~3.3]{Gavinsky2007} for all $\bf v$ of even Hamming weight:
\[
\Pr_\mathrm{M}\left[\exists {\bf u}\in\01^{n/4} \mbox{ s.t.\ }\mathbb{M}^T{\bf u}={\bf v}\right]
=\frac{\binom{n/4}{k/2}}{\binom{n}{k}}\leq (e k/2n)^{k/2},
\]
where $\mathrm{M}$ is uniform over the set of all $1/4$-matchings, and the last inequality uses the well-known bounds $\binom{n/4}{k/2}\leq (e(n/4)/(k/2))^{k/2}=(en/2k)^{k/2}$ and $\binom{n}{k}\geq (n/k)^k$. There is at most one ${\bf u}$ such that $\mathbb{M}^T {\bf u}={\bf v}$.
For ${\bf v}$ of odd Hamming weight, the probability is 0 because $\mathbb{M}^T{\bf u}$ always has twice the Hamming weight of~${\bf u}$, so its weight cannot be odd. This allows us to bound the expectation over uniform~$\mathrm{M}$:
\begin{equation}
\E_\mathrm{M}[\norm{\Delta_\mathrm{M}}_1] \leq \frac{2^{2n}}{|A|\cdot|B|}2^{-n/4}
\sum_{k=1}^{2n/4}
\sum_{{\bf v}\in\01^{n}:|{\bf v}|=k}\Pr_\mathrm{M}\left[\exists {\bf u}\in\01^{n/4} \mbox{ s.t. }\mathbb{M}^T{\bf u}={\bf v}\right]\cdot \norm{\widehat{f_1}({\bf v})}_1\cdot\norm{\widehat{f_2}({\bf v})}_1\label{eq:beforehypercon}
\end{equation}
We have the following ``hypercontractive'' inequality from~\cite[Theorem~1]{BenAroya2008} to upper bound the Fourier coefficients of a function $f:\01^n\to {\cal M}_d$: for every $p\in[1,2]$,
\begin{equation}
\sum_{{\bf v}\in\01^n}
(p-1)^{|{\bf v}|}
\norm{\widehat{f}({\bf v})}_p^2
\leq
\left(\frac{1}{2^{n}}\sum_{{\bf x}\in\01^n}\norm{f({\bf x})}_p^p\right)^{2/p}\label{eq:matrixhypercon}
\end{equation}
We will split the sum-over-$k$ of Eq.~\eqref{eq:beforehypercon} into the small and large $k$, upper bounding both parts separately. 

{\bf The case of small $k$: $k\leq q$.}
We want to use Eq.~\eqref{eq:matrixhypercon} to bound the norms of the degree-$k$ Fourier coefficients of the function $f_1$. Let $p=1+\delta$ for $\delta\in[0,1]$ to be chosen later. Lower bound the left-hand side of \eqref{eq:matrixhypercon} by restricting the sum to only the $v$'s of weight~$k$, and use $\norm{\widehat{f_1}({\bf v})}_1\leq \norm{\widehat{f_1}({\bf v})}_p$.
Upper bound the right-hand side using that $\norm{f_1({\bf x})}_p^p\leq 2^{-q}$ for all ${\bf x}\in A$ because $f_1({\bf x})$ is a $q$-qubit density matrix, and $f_1({\bf x})=0$ for ${\bf x}\not\in A$.
This gives the following bound on the degree-$k$ coefficients of~$f_1$:
\[
\sum_{{\bf v}\in\01^n:|{\bf v}|=k}\norm{\widehat{f_1}({\bf v})}_1^2\leq 
\delta^{-k}\left(\frac{|A|}{2^{n+q}}\right)^{2/(1+\delta)}
\]
Choosing $\delta=k/q$ (which is $\leq 1$ because we assumed $k\leq q$ for this case), and using $2^n/|A|\leq 2^c$ and $c\leq q$ gives
\begin{align*}
\frac{2^{2n}}{|A|^2}\sum_{{\bf v}\in\01^n:|{\bf v}|=k}\norm{\widehat{f_1}({\bf v})}_1^2
& \leq 
\delta^{-k}2^{-2q/(1+\delta)}\left(\frac{2^n}{|A|}\right)^{2\delta/(1+\delta)}\\
& \leq 2^{-2q/(1+\delta)}(4q/k)^k\\
& \leq 2^{-2q}2^{2\delta q}(4q/k)^k\\
& = 2^{-2q}(16q/k)^k.
\end{align*}
We have the analogous bound for the degree-$k$ Fourier coefficients of~$f_2$, with $B$ instead of $A$.

Using Cauchy-Schwarz with these Fourier bounds for $f_1$ and $f_2$, we bound the part of the sum of \eqref{eq:beforehypercon} where $k\leq q$ by
\begin{align*} 
&\frac{2^{2n}}{|A|\cdot|B|}2^{-n/4}
\sum_{k=1}^{q}
\frac{\binom{n/4}{k/2}}{\binom{n}{k}}
\sqrt{\sum_{{\bf v}\in\01^{n}:|{\bf v}|=k}\norm{\widehat{f_1}({\bf v})}_1^2}\,\sqrt{\sum_{{\bf v}\in\01^{n}:|{\bf v}|=k}\norm{\widehat{f_2}({\bf v})}_1^2}\\
&=2^{-n/4}
\sum_{k=1}^{q}
\frac{\binom{n/4}{k/2}}{\binom{n}{k}}
\sqrt{\frac{2^{2n}}{|A|^2}\sum_{{\bf v}\in\01^{n}:|{\bf v}|=k}\norm{\widehat{f_1}({\bf v})}_1^2}\,\sqrt{\frac{2^{2n}}{|B|^2}\sum_{{\bf v}\in\01^{n}:|{\bf v}|=k}\norm{\widehat{f_2}({\bf v})}_1^2}\\
&\leq 2^{-2q-n/4}
\sum_{{\rm even\ }k=2}^{q} 
\frac{\binom{n/4}{k/2}}{\binom{n}{k}}(16q/k)^k\\
&\leq 2^{-2q-n/4}
\sum_{{\rm even\ }k=2}^{q} (ek/2n)^{k/2}(16q/k)^k\\
&\leq 2^{-2q-n/4}
\sum_{{\rm even\ }k=2}^{q}O(q/\sqrt{kn})^k\\
&\leq 2^{-2q-n/4}\,O(q^2/n),
\end{align*}
where the last inequality is because the sum in the penultimate line is upper bounded by a geometric sum (by dropping the $\sqrt{k}$ from the denominator and assuming $n$ is sufficiently large), which is dominated by the term with $k=2$.

{\bf The case of large $k$: $k>q$.}
Norm-monotonicity and Eq.~\eqref{eq:matrixhypercon} with $p=2$ implies 
\[
\sum_{v\in\01^n}\norm{\widehat{f_1}({\bf v})}_1^2\leq \sum_{{\bf v}\in\01^n}\norm{\widehat{f_1}({\bf v})}_2^2\leq \frac{1}{2^{n}}\sum_{{\bf x}\in\01^n}\norm{f_1(x)}_2^2\leq \frac{|A|}{2^{q+n}},
\]
where the last inequality again uses that $\norm{f_1({\bf x})}_p^p\leq 2^{-q}$ for all ${\bf x}\in A$ because $f_1({\bf x})$ is a $q$-qubit density matrix, and $f_1({\bf x})=0$ for ${\bf x}\not\in A$.
We have a similar inequality for $f_2$.

Define $g(k)=\binom{n/4}{k/2}/\binom{n}{k}$. 
\cite{Gavinsky2007} show that $g$ is monotone nonincreasing for $k\in\{q,\ldots,n/2\}$.
Together with Cauchy-Schwarz, this allows us to upper bound the part of the sum of \eqref{eq:beforehypercon} with $k>q$ by 
\begin{align*}
\\
&\leq \frac{2^{2n}}{|A|\cdot|B|}2^{-n/4}
g(q)
\sqrt{\sum_{{\bf v}\in\01^{n}}\norm{\widehat{f_1}({\bf v})}_1^2}\sqrt{\sum_{{\bf v}\in\01^{n}}\norm{\widehat{f_2}({\bf v})}_1^2}\\
&\leq \frac{2^{n}}{\sqrt{|A|\cdot|B|}}2^{-q-n/4}g(q)\\
&\leq \frac{2^{n}}{\sqrt{|A|\cdot|B|}}2^{-q-n/4}(eq/2n)^{q/2}\\
&\leq 2^{c-q-n/4}(eq/2n)^{q/2}\\
& \leq 2^{-2q-n/4}\cdot q^2/n,
\end{align*}
where the last inequality uses $c\leq q$ and $eq/2n \leq 2^{-1/4}$ for sufficiently large $n$ (recall that we assumed $q\leq \sqrt{n}$). 
This concludes the upper bound for the large-$k$ case.

The trace-norm (=sum of singular values, which is twice the total-variation distance) of an $m$-qubit matrix is $2^{m}$ times its Schatten 1-norm, hence combining the upper bounds for the small and the large~$k$, we have
\[
\E_\mathrm{M}[\norm{\Delta_{\mathrm{M}}}_{tr}] =2^{2q+n/4}\E_\mathrm{M}[\norm{\Delta_\mathrm{M}}_1]\leq O(q^2/n). 
\]
This completes the proof. 
\end{proof}

The two main lower bound claims in our paper now follow as immediate corollaries:

\medskip

\noindent {\bf Proof of Theorem~\ref{theo2}:} 
It suffices to prove the theorem for $m=2$, since for $m\geq 3$ we can fix the inputs for $m-1$ senders to the all-0 string $0^n$ to reduce to the $m=2$ case.
Assume without loss of generality that the message lengths from the two Alices are equal to the same~$q$.

By choosing $A=B=\01^n$ (so $c=0$), $\varepsilon$ a small constant, $q=C\sqrt{n}$ for small constant~$C$, we obtain a lower bound of $q=\Omega(\sqrt{n})$ for the bounded-error quantum communication complexity of the two-Alice version of Boolean Hidden matching, because a protocol that learns $b$ (with success probability $\geq 2/3$) from $f_1({\bf x}^1),f_2({\bf x}^2),\mathrm{M}$ would be able to (with large probability) distinguish $Z_{\mathrm{M}}({\bf x}^1,{\bf x}^2)$ from uniformly random, violating Theorem~\ref{th:mainlowerbound}. Here the two Alices (and Bob) are allowed to share classical randomness but not entanglement. Shared randomness is not dealt with by Theorem~\ref{th:mainlowerbound} itself, but it is easy to see that shared randomness between the two Alices and/or Bob cannot improve their success probability under a fixed input distribution, since one can always fix the ``best'' shared random string w.r.t.\ that input distribution.

\medskip
\noindent {\bf Proof of part (1) of Theorem~\ref{theo3}:} 
By choosing $c\leq q=C\sqrt{\varepsilon n}$ for sufficiently small constant $C$, we also obtain the proof that $Z({\bf x}^1,{\bf x}^2,\mathrm{M})$ is an $\varepsilon$-secure strong seeded two-source extractor against two unentangled quantum adversaries who store $q=O(\sqrt{\varepsilon n})$ qubits about ${\bf x}^1$ and ${\bf x}^2$, respectively, assuming ${\bf x}^1,{\bf x}^2$ have min-entropy at least $n-c=n-O(\sqrt{\varepsilon n})$. Our Theorem~\ref{th:mainlowerbound} only handles the special case where the distributions of ${\bf x}^1,{\bf x}^2$ are so-called flat sets (uniform over sets $A,B$ of size $\geq 2^{n-c}$, respectively), but Chor and Goldreich~\cite{Chor1988} showed that every distribution of min-entropy $\geq n-c$ is a convex combination of such flat sets, so our proof extends immediately to the general case.

\end{document}